\documentclass[letterpaper, 10 pt, conference]{ieeeconf}

\usepackage{graphicx}
\IEEEoverridecommandlockouts                 
\usepackage{amsmath,amssymb,amsfonts}
\usepackage{subfigure}
\usepackage{graphicx}
\graphicspath{{./IMG/}}
\usepackage{epstopdf}
\usepackage{cite}
\usepackage{hyperref}
\usepackage{epstopdf,color}
\usepackage{textcomp}
\usepackage{enumerate}

\usepackage{color,array}
\usepackage{comment}
\usepackage{amsmath}
\usepackage{etoolbox}
\usepackage{algorithm,algorithmic}
\usepackage{mathrsfs}
\AtBeginEnvironment{algorithm2e}{\algoequations}
\AtEndEnvironment{algorithm2e}{\restoreequations}
\newcounter{algosavedequation}
\newcommand{\algoequations}{%
  \setcounter{algosavedequation}{\value{equation}+1}%
  \setcounter{equation}{0}%
  \renewcommand{\theequation}{\arabic{algosavedequation}\alph{equation}}%\thealgocf
}
\newcommand{\restoreequations}{%
  \setcounter{equation}{\value{algosavedequation}}%
}

\allowdisplaybreaks[3]

\newcommand{\longthmtitle}[1]{\mbox{}\emph{(#1):}}

\newtheorem{problem}{Problem}

\newtheorem{remark}{Remark}

\newtheorem{proposition}{Proposition}
\newtheorem{lemma}{Lemma}

\newcommand{\real}{\mathbb{R}}

\newcommand{\naturalpos}{\mathbb{N}_{>0}}

\newcommand*{\QEDBL}{\hfill\ensuremath{\blacksquare}}% Black square
\newcommand\oprocendsymbol{\hbox{$\square$}}
\newcommand\oprocend{\relax\ifmmode\else\unskip\hfill%
\fi\oprocendsymbol}

\DeclareMathAlphabet{\mymathbb}{U}{BOONDOX-ds}{m}{n}

\newcommand{\setdef}[2]{\{#1 \; : \; #2\}}

\newcommand{\bc}{{\mathbf{c}}}
\newcommand{\bx}{{\mathbf{x}}}

\newcommand{\bu}{{\mathbf{u}}}

\newcommand{\Lc}{\mathcal{L}}

\newcommand{\Kc}{{\mathcal{K}}}
\newcommand{\Cc}{{\mathcal{C}}}

\newcommand{\norm}[1]{\Vert #1 \Vert}

\renewcommand{\baselinestretch}{.99}

\begin{document}

\title{\bf Safe Policy Optimization via  \\ Control Barrier Function-based Safety Filters
}

\author{Yiting Chen \quad Pol Mestres  \quad Emiliano Dall'Anese \quad Jorge Cort\'{e}s
\thanks{Y. Chen and E. Dall'Anese are with the Department of Electrical and Computer Engineering at Boston University; P. Mestres is with the Department of Mechanical Engineering at the California Institute of Technology; J. Cort\'{e}s is with the Department of Mechanical and Aerospace Engineering at the University of California San Diego.}
\thanks{This work was supported by the AFOSR Award FA9550-23-1-0740.}
}

\maketitle

\begin{abstract}%
Control barrier function (CBF)-based safety filters provide a systematic way to enforce state constraints, but they can significantly alter the closed-loop dynamics induced by a nominal, stabilizing controller. In particular, the resulting safety-filtered system may exhibit undesirable behaviors including limit cycles, unbounded trajectories, and undesired equilibria. This paper develops a policy optimization framework to maximally enhance the stability properties of safety-filtered controllers. Focusing on linear systems with linear nominal controllers, we jointly parameterize the nominal feedback gain and safety-filter components, and optimize them using trajectory-based objectives computed from closed-loop rollouts. To ensure that the nominal controller remains stabilizing throughout training, we encode Lyapunov-based stability conditions as smooth scalar constraints and enforce them using robust safe gradient flow. This guarantees feasibility of the stability constraints along the optimization iterates and therefore avoids instability during training. Numerical experiments on obstacle-avoidance problems show that the proposed approach can remove asymptotically stable undesired equilibria and improve convergence behavior while maintaining forward invariance of the safe set.

\end{abstract}
%
% \marginJC{Since we only deal with linear systems, let's make sure we reflect that in the abstract/contributions/etc. Reviewers get rightfully frustrated if you overpromise/overstate and then don't deliver.}
%

%\begin{keywords}%

%\end{keywords}

%%%%%%%%%%%%%%%%%%%%%%%%%%%%%%%%%%
\section{Introduction}
\label{sec:introduction}

Ensuring safety while achieving high-performance control is a fundamental requirement in modern autonomous and cyber-physical systems, including robotics, transportation, and energy systems. 
In many applications, safety is formalized as the forward invariance of a prescribed set of ``safe'' states~\cite{ADA-SC-ME-GN-KS-PT:19}. Control barrier functions (CBFs) provide a systematic framework for enforcing such safety constraints in control-affine systems~\cite{ADA-XX-JWG-PT:17, LW-ADA-ME:17}, commonly through \emph{safety filters} that, given a nominal controller, minimally modify to ensure safety.
However, the interaction between nominal controllers and CBF-based safety filters can significantly alter the properties of the resulting closed-loop dynamics. In particular, the safety filter may degrade the stability properties of the nominal controller. For instance, even when the nominal controller ensures global asymptotic stability of the origin, the safety-filtered system may exhibit undesirable behaviors such as asymptotically stable undesired equilibria, limit cycles, or unbounded trajectories~\cite{chen2024characterization, PM-YC-EDA-JC:25-jnls, MFR-APA-PT:21}. 
These observations highlight the need for systematic methods to design controllers with guarantees for both safety and stability properties. 
% This paper aims to develop a systematic approach to jointly design nominal controllers and safety filters to improve both safety and stability.

\emph{Literature Review:}
Recent works have used Lyapunov-based ideas to stabilize learned controllers, including joint controller--Lyapunov design~\cite{long2025certifying,wu2023neural}, Lyapunov learning from demonstrations~\cite{mittal2020neural}, and Lyapunov-regularized policy optimization~\cite{hejase2023lyapunov}. However, these works do not address the interaction between nominal controllers and CBF-based safety filters. In parallel, recent results have shown that safety filters can significantly alter the qualitative behavior of the closed-loop system: even when the nominal controller is asymptotically stabilizing, the filtered dynamics may exhibit undesired equilibria, limit cycles, or unbounded trajectories~\cite{chen2024characterization,PM-YC-EDA-JC:25-jnls,MFR-APA-PT:21}. Moreover, while some of these obstructions are structural, the choice of nominal controller and safety-filter design can still strongly affect their impact on the closed-loop dynamics~\cite{chen2024characterization,PM-YC-EDA-JC:25-jnls}. This motivates the development of systematic methods to optimize nominal controllers and safety parameters in order to improve the behavior of safety-filtered systems.
%
% \marginJC{"to improve both safety and stability" feels a bit generic. We don't improve safety (safety filters already give you that), no? And the lit review is only focused on the learning side (of course we need it), but it does nto say that we've identified hard obstructions that cannot be overcome with safety filters no matter what, but that also has a positive spin, which is to try to optimize the nominal controller design to minimize the effect of such obstructions. In a similar vein, in our paper in Automatica, we show that these properties remain no matter what CBF/alpha pair, but on the other hand, the choice makes a difference in making the extent to which those properties affect the dynamics (e.g., size of region of attraction of a stable undesired equilibrium).}
%

\emph{Statement of Contributions:}
This paper develops a systematic framework to improve the dynamical behavior of safety-filtered controllers through optimization of the nominal controller and safety-filter parameters. The main contributions are as follows:

 \emph{(c1)} We formulate a trajectory-based optimization problem for safety-filtered controllers in which the nominal controller, the CBF-related class-$\mathcal K_\infty$ function, and the safety-filter weighting matrix are jointly parameterized. This formulation is designed to shape the closed-loop dynamics induced by safety filtering, with the goal of reducing the effect of undesired equilibria and improving convergence to the desired equilibrium.

 \emph{(c2)} Focusing on linear systems with linear nominal controllers, we encode nominal closed-loop stability through Lyapunov-based matrix inequalities and reformulate these conditions as smooth scalar constraints using leading principal minors. This yields an optimization problem amenable to gradient-based methods while preserving an explicit characterization of stabilizing nominal controllers.

 \emph{(c3)} We develop a training procedure based on robust safe gradient flow and rollout-based gradient estimates. Proposition~\ref{prop:constraint-satisfaction-convergence} shows that, when initialized from a stabilizing nominal controller, all iterates remain feasible and therefore keep the nominal controller stabilizing throughout training. The result also quantifies how gradient approximation errors affect convergence to a neighborhood of a KKT point.

 \emph{(c4)} We illustrate the proposed framework on representative obstacle-avoidance examples, including environments with multiple constraints and complex safe-set geometry. The numerical results show that the optimized controllers improve convergence behavior and can eliminate asymptotically stable undesired equilibria while preserving forward invariance of the safe set.

\section{Preliminaries}

\subsection{Notation} 
We denote by $\real$ and $\naturalpos$ the set of real and positive integers, respectively. 
Vectors are written in boldface, while scalar quantities are not. 
Given $\bx=(x_1,...,x_n)\in\real^n$ and $p\in\mathbb{Z}_{>0}$, $\|\bx\|_p$ denotes its $\ell_p$ norm. For a square matrix $G\in\real^{m\times m}$, the $i$-th leading principal minor is defined as the determinant of $G_{(1:i,\,1:i)}$, where $G_{(1:i,\,1:i)}$ denotes the submatrix formed by the first $i$ rows and columns of $G$. If $G$ is positive definite, we define $\norm{\bu}_G := \sqrt{\bu^\top G \bu}$. Given a continuously differentiable function $h:\real^n\to\real$, $\nabla h(\bx)$ denotes its gradient at $\bx$. 
A function $\alpha:\real\to\real$ is of extended class $\mathcal{K}_\infty$ if it is strictly increasing, $\alpha(0)=0$, and $\lim_{s\to\pm\infty}\alpha(s)=\pm\infty$.
A function $f:\real^n\to\real$ is $L$-smooth for some constant $L>0$ if $\norm{\nabla f(x)-\nabla f(y)} \leq L\norm{x-y}$.

%%%%%%%%%%%%%%%%%%%%%%%%%%%%%%%%%%
\subsection{Safety filters for control-affine systems}\label{sec:prelim}

Consider the control-affine system
\begin{align}\label{eq:prelim-control-affine}
    \dot{\bx} = f(\bx) + g(\bx)\bu,
\end{align}
where $f:\real^n\to\real^n$ and $g:\real^n\to\real^{n\times m}$ are locally Lipschitz, $\bx\in\real^n$ is the state, and $\bu\in\real^m$ is the control input.

Let $h:\real^n\to\real$ be a continuously differentiable function, and define the set
\[
\Cc := \{\bx\in\real^n : h(\bx) \geq 0\}.
\]
The function $h$ is a \emph{control barrier function} (CBF) for the set $\Cc$ if there exists an extended class $\mathcal{K}_\infty$ function $\alpha$ such that, for all $\bx\in\Cc$, there exists $\bu\in\real^m$ satisfying
\begin{equation}
\label{eq: cbf-condition}
\nabla h(\bx)^\top (f(\bx) + g(\bx)\bu) + \alpha(h(\bx)) \geq 0.  
% \hfill \Box
\end{equation}

The main feature of CBFs is that they enable the design of safe controllers.
Indeed, any locally Lipschitz controller satisfying~\eqref{eq: cbf-condition} for all $\bx\in\Cc$ renders $\Cc$ forward invariant~\cite[Corollary 2]{ADA-XX-JWG-PT:17}.
%
% \marginJC{Sloppy here. We're thinking of state feedback, not open-loop inputs. ANd we need to ask for locally Lipschitzness, otherwise sols of closed-loop system might not be unique and Nagumo's theorem does not apply.}
%
A common method to design such controller is through \textit{safety filters}.
Given a nominal controller $k:\real^n\to\real^m$ that stabilizes a desired equilibrium point (taken to be the origin henceforth without loss of generality),
safety filters typically search for a controller that minimally modifies $k$ while satisfying the CBF condition~\eqref{eq: cbf-condition}.
This controller can be obtained at every state $\bx\in\Cc$ by solving the following quadratic program (QP):
\begin{align}\label{eq:prelim-safety-filter}
    v(\bx) &= \arg\min_{\bu\in\real^m} \frac{1}{2}\norm{\bu - k(\bx)}^2_{G(\bx)} \\
    &\text{s.t.} \quad \nabla h(\bx)^\top (f(\bx)+g(\bx)\bu) + \alpha(h(\bx)) \geq 0, \nonumber
\end{align}
where $G:\real^n\to\real^{m\times m}$ is continuously differentiable and positive definite for all $\bx\in\real^n$ (this function prescribes the cost function weights at every state).
%
% \marginJC{What is $G(x)$? Never introduced. You need to also assume something on the dependency of $G$ on $x$ (continuous?).}
%
Since the objective function of~\eqref{eq:prelim-safety-filter} is strictly convex, $v(\bx)$ is a singleton for each $\bx\in\Cc$. Furthermore, if Slater's condition holds for~\eqref{eq:prelim-safety-filter} for each $\bx\in\Cc$, then $v$ is locally Lipschitz~\cite{MA-NA-JC:25-tac}, and renders $\Cc$ forward invariant.
We also note that $v$ can be expressed in closed-form~\cite{PM-YC-EDA-JC:25-jnls}.
Indeed, let $\eta(\bx) = \nabla h(\bx)^\top (f(\bx)+g(\bx)k(\bx)) + \alpha(h(\bx))$. Then,
\begin{align}\label{eq:v-expression}
    v(\bx) = \begin{cases}
        k(\bx), & \text{if } \eta(\bx) \geq 0, \\
        k(\bx) + \bar{u}(\bx), & \text{if } \eta(\bx) < 0,
    \end{cases}
\end{align}
where
\[
\bar{u}(\bx) := -\frac{\eta(\bx)\, G(\bx)^{-1} g(\bx)^\top \nabla h(\bx)}
{\| g(\bx)^\top \nabla h(\bx) \|_{G(\bx)^{-1}}^2}.
\]

% \red{[ED:  (4) is incorrect. it shoudl be
% \begin{align}\label{eq:v-expression}
%     v(\bx) = \begin{cases}
%        k(\bx), & \text{if } \eta(\bx) \geq 0, \\
%         k(\bx) + \bar{u}(\bx), & \text{if } \eta(\bx) < 0,
%     \end{cases}
% \end{align}
% with $\bar{u}(\bx)$ written as above]
% }

Although the safety filter guarantees forward invariance of $\Cc$ and admits an efficient implementation, it does not, in general, preserve the stabilization properties of the nominal controller. 
In particular, even if the nominal controller $\bu = k(\bx)$ is designed so that the system $\dot \bx = f(\bx) + g(\bx)k(\bx)$ renders the origin globally asymptotically stable, the closed-loop system obtained by applying the safety filter may fail to retain this property. 
The results in~\cite{chen2024characterization,PM-YC-EDA-JC:25-jnls,MFR-APA-PT:21,tan2024undesired} show that the resulting closed-loop system $\dot \bx = f(\bx) + g(\bx)v(\bx)$ can exhibit undesirable dynamical behaviors, such as the presence of additional equilibria (whose number, location, and stability depend on the nominal controller), unbounded trajectories, or limit cycles. 
% As discussed earlier, the use of a safety filter can significantly alter the dynamical behavior of the closed-loop system.  In particular, the resulting system may exhibit additional equilibria whose number, location, and stability depend on the nominal controller. 
Furthermore, as observed in~\cite{chen2024characterization,PM-YC-EDA-JC:25-jnls},
%
% \marginJC{SHouldn't we also refer to PM-YC-EDA-JC:25-jnls here?}
%
different nominal controllers and different choices of the CBF
can lead to different behaviors when combined with a safety filter. 
In favorable cases, the closed-loop system obtained from the safety filter is also globally asymptotically stable,
or has only one unstable undesired equilibrium (this is the best dynamical behavior one can expect in cases where the safe set is not simply connected, due to topological obstructions\cite{DEK:87-icra}).
In less favorable cases, the safety filter may introduce multiple undesired equilibria, some of which may even be asymptotically stable. However, as shown in~\cite{chen2024characterization}, the choice of nominal controller or CBF can be used to reduce the region of attraction of the asymptotically stable undesired equilibria.

% Moreover, due to topological obstructions \cite{DEK:87-icra}, it is in general impossible to design a locally Lipschitz controller that renders all trajectories convergent to the desired equilibrium while preserving safety. 
% In addition, the choice of the barrier function $h$ and the function $\alpha$ can affect the size of the region of attraction of the desired equilibrium \cite{chen2024equilibria}. 

%
%\marginJC{I think we could (should!) do a better job here. "favorable" is loose. Plus, we know much more: e.g., the undesirable things that do not change qualitatively even if you consider other CBFs (with the same 0-superlevel set) or other class $K$ functions. Then you need to explain that quantitatively, things may change. Maybe a figure/plot illustrating this point for 2 different $(h_1,\alpha_1)$ and $(h_2,\alpha_2)$ pairs would be most helpful. This sets the stage for the problem statement in the next section.}
%

\section{Problem Statement}

Motivated by the considerations of Section~\ref{sec:prelim} regarding the impact of safety filters on dynamical behavior, we seek to determine a nominal controller that leads to a closed-loop system with favorable dynamical properties. By this, we mean ensuring the number of undesired equilibria is minimized, preventing the emergence of asymptotically stable undesired equilibria, and maximizing the region of attraction of the desired equilibrium. We formalize this search 
by considering a parametric family of nominal controllers, candidate CBFs (i.e., continuously differentiable functions whose $0$-superlevel set is $\Cc$), cost function weighting functions, and class $\Kc$ functions.
Specifically, we let $\theta\in\real^d$ be a parameter vector, and define
\begin{align*}
&k_\theta:\real^n\to\real^m, &  & G_\theta:\real^n\to\real^{m\times m} ,
\\
&\alpha_\theta:\real\to\real,&  & h_\theta:\real^n\to\real,
\end{align*}
where $k_{\theta}$ is locally Lipschitz,
%and globally asymptotically stable, (no, we impose it later as a constraint) 
$G_{\theta}$ is continuously differentiable and positive definite for all $\bx\in\Cc$, $\alpha_{\theta}$ is an extended class $\Kc_{\infty}$ function, and $h_{\theta}$ is continuously differentiable and its zero-superlevel set is $\Cc$.

Given $\theta$, the corresponding safety-filtered controller is
\begin{align}\label{eq:prelim-parameterized-filter}
    \bu^*_\theta(\bx)
    &= \arg\min_{\bu\in\real^m} \frac{1}{2}\norm{\bu - k_\theta(\bx)}_{G_\theta(\bx)}^2 \\
    &\text{s.t.} \quad \nabla h_{\theta}(\bx)^\top (f(\bx)+g(\bx)\bu) + \alpha_\theta(h_\theta(\bx)) \geq 0. \nonumber
\end{align}
%
% \marginJC{Here, what are the regularity properties of these maps? Also, all $h_\theta$ should have the same 0-superlevel set $C$, no? All $\alpha_\theta$ should be class BLA functions, etc., no?}
%
which induces the 
% {\color{cyan} I would add $\theta$ dependency for $h$ here as well because it links well with our Automatica paper. However, in the sims we don't need to do it.}
closed-loop system
\begin{align}\label{eq:prelim-closed-loop}
    \dot{\bx} = f(\bx) + g(\bx)\bu^*_\theta(\bx).
\end{align}

To capture the desirable dynamical objectives, we consider a trajectory-based cost function defined over a finite time horizon. Given a finite horizon $T>0$ and an initial condition $\bx_0$, let $\bx_\theta(t;\bx_0)$ denote the solution of~\eqref{eq:prelim-closed-loop} at time $t$. 
We define the cost function
\begin{align}\label{eq:prelim-cost-single}
    \mathcal{L}(\theta,\bx_0)
    = \phi(\bx_\theta(T;\bx_0) )
    + \lambda \int_0^T \psi(\bx_\theta(t;\bx_0)) dt,
\end{align}
for some $\lambda>0$. The terminal cost $\phi(\bx_\theta(T;\bx_0))$ penalizes trajectories that fail to approach the desired equilibrium over the time horizon. 
In particular, it encourages the state to be close to the origin at time $T$, thereby promoting convergence to the desired equilibrium. The running cost $\int_0^T \psi(\bx_\theta(t;\bx_0)) dt$ penalizes deviations from the desired equilibrium along the trajectory. 
As a result, it discourages trajectories that remain far from the origin for extended periods of time, and promotes faster convergence. The hyperparameter $\lambda$ balances the relative importance of the terminal and running costs.

Given a distribution $\xi$ over initial conditions, we consider the expected cost
\begin{align}\label{eq:prelim-cost}
    \mathcal{L}(\theta) = \int_{\real^n} \mathcal{L}(\theta,\bx_0)\, \xi(\bx_0)\, d\bx_0,
\end{align}
and the optimization problem
\begin{align}\label{eq:prelim-optimization}
    \min_{\theta\in\real^d} \mathcal{L}(\theta).
\end{align}
Qualitatively, problem~\eqref{eq:prelim-optimization} seeks parameters $\theta$ that lead to closed-loop trajectories that approach the desired equilibrium for initial conditions from a given distribution.
In particular, it promotes closed-loop systems with favorable dynamical properties, such as a large region of attraction of the desired equilibrium and the absence of asymptotically stable undesired equilibria.

Our goal is to design a safe controller with favorable dynamical properties by solving~\eqref{eq:prelim-optimization}.
In particular, we aim to solve the following problem.

\vspace{0.1cm}

\begin{problem}\label{problem:learning-cbf}
Given system~\eqref{eq:prelim-control-affine}, a safe set $\Cc$, and the parameterized safety filter~\eqref{eq:prelim-parameterized-filter},
design an algorithm that solves~\eqref{eq:prelim-optimization}.
Additionally, in order to ensure that the origin of the closed-loop system remains asymptotically stable (even if the algorithm is terminated after a finite number of iterations), ensure that all the algorithm iterates yield a stabilizing nominal controller.
\hfill $\Box$
\end{problem}
%
% \marginJC{This environment does not say anything about "investigate the dynamical properties of the resulting closed-loop system". Maybe more importantly, there is no constraint on having the nominal controller be stabilizing. I know we address this in the next section, but shouldn't we discuss/introduce the problem already with that. It could replace Problem 1, or maybe just discuss the shortcoming of addressing Problem 1, leading to the formulation of a Problem 2 (the one w/ stabilizing nominal controllers) which is the one that we really address in the paper.}
%

%%%%%%%%%%%%%%%%%%%%%%%%%%%%%%%%%%%%%

\section{Enforcing stability constraints}

The formulation~\eqref{eq:prelim-optimization} in Section~\ref{sec:prelim} prioritizes safety by construction, as the CBF-based filter guarantees forward invariance of the set $\Cc$. 
However, it does not guarantee that the nominal controller is stabilizing. 
As mentioned in Problem~\ref{problem:learning-cbf}, to address this our goal is to select parameters $\theta$ for which the nominal controller renders the origin asymptotically stable.

In this section we explain how to enforce such stability constraints on the nominal controller.
We focus on linear systems of the form
\begin{align}\label{eq:linear-system}
    \dot{\bx} = A\bx + B\bu,
\end{align}
($f(x) = Ax$ and $g(x) = B)$
and consider linear nominal controllers $k_\theta(\bx) = -K\bx$,
where $K\in\real^{m\times n}$ is part of the parameter vector $\theta$, i.e., $\theta = (K,\bar{\theta})$, for some other parameters $\bar{\theta}$ parameterizing the rest of the functions $G_{\bar{\theta}}$, $\alpha_{\bar{\theta}}$, $h_{\bar{\theta}}$.

% \red{[ED: there is a sign inconsistency here: You write the nominal controller as $k_\theta(\bx) = K\bx$, but then enforce stability of $A-BK$. The controller should be parametrized as $k_\theta(\bx) = - K\bx$ to be consistent.]}

We seek to enforce that the matrix $A-BK$ is Hurwitz. 
A standard approach is to introduce a Lyapunov function $V(\bx) = \bx^\top P \bx$, with $P\in\real^{n\times n}$ positive definite, and impose the condition
\begin{align}\label{eq:lyapunov-condition}
    (A-BK)^\top P + P(A-BK) \prec 0.
\end{align}

The condition~\eqref{eq:lyapunov-condition} is bilinear in $(K,P)$ and therefore non-convex. 
Using the standard change of variables $Y = KP$ and $Q = P^{-1}$, we obtain the equivalent linear condition
\begin{align}
    AQ + QA^\top - BY - Y^\top B^\top \prec 0, \quad Q \succ 0.
\end{align}

In the remainder of the paper, we set $h_\theta \equiv h$, and only optimize over the remaining parameters. For tractability, we adopt simple parametric forms for the functions $\alpha_\theta$ and $G_\theta$. Specifically, we parameterize the extended class $\mathcal{K}_\infty$ function as 
$\alpha_\theta(s) = \alpha s$, where $\alpha > 0$ is a scalar decision variable. 
This linear parameterization satisfies the $\mathcal{K}_\infty$ properties. Moreover, we note that $G_\theta$ affects $\mathcal{L}(\theta)$ only through $\bu^*_\theta$ 
(as in \eqref{eq:prelim-parameterized-filter}), and $\bu^*_\theta$ depends on $G_\theta$ 
only via its inverse (cf.~\eqref{eq:v-expression}). 
Therefore, we parameterize $G_\theta$ as $G_\theta = R^{-1}$, where $R \succ 0$. Accordingly, we consider a modified version of \eqref{eq:prelim-optimization} that incorporates the above constraints:
%
% \marginJC{THis is no longer (9), b/c we've added constraints that were not there. I think it's better if we say that we consider a modified version of (9) that incorporates stability constraints on the nominal controller.}
%
\begin{align}\label{eq:non-convex-lmi}
    &\min_{\theta = (Y,Q,R,\alpha)\in\real^d} \mathcal{L}(\theta) \\
    &\text{s.t.} \quad AQ + QA^\top - BY - Y^\top B^\top \prec 0, \nonumber\\
    &\phantom{\text{s.t.}} \quad Q \succ 0, \nonumber\\
     &\phantom{\text{s.t.}} \quad R \succ 0, \nonumber\\
    &\phantom{\text{s.t.}} \quad \alpha > 0, \nonumber
\end{align}
% \marginED{ED: It is unclear to me how this problem is formulated. In particular, $\alpha$ appears as an independent variable in an independent constraint. There should be the dynamics as a constraint as well as the filter, right?}
%
% \marginJC{Should we add 1 more constraint to make it explicit? Something like "under the closed-loop dynamics~\eqref{eq:prelim-closed-loop}"?}
%
where recall that the cost function is computed using the closed-loop dynamics~\eqref{eq:prelim-closed-loop}. In fact, the optimization variables $\theta = (Y,Q,R,\alpha)$ affect $\mathcal{L}(\theta)$ through the following chain of dependencies: 
the matrices $Y$ and $Q$ determine the nominal controller $k_\theta$ through the relation $K = YQ^{-1}$, while the scalar $\alpha$ and the matrix $R$ determine the parameterized functions $\alpha_\theta(s)=\alpha s$ and $G_\theta = R^{-1}$, respectively. 
These quantities jointly determine the filtered control input $\bu_\theta^*(\bx)$ through \eqref{eq:prelim-parameterized-filter}. 
The control input $\bu_\theta^*(\bx)$ in turn defines the closed-loop dynamics~\eqref{eq:prelim-closed-loop}, which determine the single-trajectory cost in \eqref{eq:prelim-cost-single}.  Finally, the expected loss $\mathcal{L}(\theta)$ is obtained by averaging this cost over the trajectories induced by the resulting closed-loop system.

The first three constraints in~\eqref{eq:non-convex-lmi} are linear matrix inequalities (LMI). 
However, directly handling this constraint within a gradient-based learning framework is challenging, as it requires projecting onto the cone of positive definite matrices or solving a semidefinite program at each iteration. To obtain a representation that is more amenable to gradient-based optimization, 
we seek to replace the matrix inequality constraints with equivalent scalar algebraic conditions. To this end, we exploit a classical equivalent characterization of positive definiteness based on principal minors.

\vspace{0.1cm}
\begin{lemma}[{\cite[Theorem 7.2.5]{horn2012matrix}}]\label{lem:principal minors}
Let $ M\in \real^{n\times n}$ be a symmetric matrix. Then $M$ is positive definite if and only if all leading principal minors of $M$ are positive, i.e., $\det(M_{(1:i,\,1:i)}) > 0$, for all $i\in[n]$. \hfill$\Box$
\end{lemma}
\vspace{0.1cm}

By Lemma \ref{lem:principal minors}, we can equivalently replace each LMI in \eqref{eq:non-convex-lmi} with $n$ inequalities based on leading principal minors. Specifically, we write
\[
\Delta(Q,Y) := -AQ - QA^\top + BY + Y^\top B^\top,
\]
and then \eqref{eq:non-convex-lmi} can be equivalently converted to
\begin{align}\label{eq:non-convex-reformulation}
    &\min_{\theta = (Y,Q,R,\alpha)\in\real^d} \mathcal{L}(\theta) \\
    &\text{s.t.} \quad \det(\Delta_{(1:i,\,1:i)}(Q,Y)) > 0, \quad i\in[n], \nonumber\\
     &\phantom{\text{s.t.}} \quad\det(Q_{(1:i,\,1:i)}) > 0, \quad i\in[n],\nonumber\\
     &\phantom{\text{s.t.}} \quad\det(R_{(1:i,\,1:i)}) > 0, \quad i\in[m],\nonumber\\
     &\phantom{\text{s.t.}} \quad \alpha>0 \nonumber,
\end{align}
where we recall that $\Delta_{(1:i,\,1:i)}(Q,Y)$ is the submatrix formed by the first $i$ rows and columns of $\Delta(Q,Y)$. Similarly, $Q_{(1:i,\,1:i)}$ and $R_{(1:i,\,1:i)}$ denote the leading principal submatrices of $Q$ and $R$, respectively. In this form, both the objective function and the constraints are non-convex, but the constraints are expressed in a standard form that can be handled by gradient-based methods.

A critical challenge in solving~\eqref{eq:non-convex-reformulation} 
% standard constrained optimization methods typically guarantee feasibility only asymptotically.
is that, violating the constraints may lead to ill-defined or unbounded values of the objective function $\Lc(\theta)$, as the underlying cost becomes arbitrarily large or ill-conditioned when the closed-loop system loses stability. 
As a result, temporary infeasibility during the optimization process can severely disrupt the algorithm and make it difficult to recover. Therefore, it is critical to employ an optimization scheme that maintains the feasibility of the constraints at every iteration, specifically one that does not leave the set of nominal stabilizing controllers.
This also ensures that even if the algorithm is terminated after a finite number of iterations, the resulting parameter is feasible.
This requirement motivates the algorithmic approach adopted in the next section.

%
% \marginJC{Before we get to the algorithm implementation in the next section, something that does not come across very clearly from the exposition of this section is the problem one would face when solving (13) or (14) with any optimization method that only guarantees convergence to optimizer asymptotically: typically, that means satisfaction of the constraints in the limit. However, b/c of the nature of problem (13) and (14), stepping outside the feasible set would lead to infinity values in the cost function, messing up the algorithm and leaving no clear option for how to restart. Hence, it is important to remain feasible throughout the algorithm evolution, no? Hence our algorithmic implementation.}
% %

% For completeness, we briefly comment on an alternative approach to enforcing stability.

\vspace{0.1cm}

\begin{remark}\longthmtitle{Alternative approach to enforcing stability}
An alternative approach to enforcing stability is to require that the spectral abscissa of $A - BK$ is negative. 
However, this results in a non-smooth constraint, which makes it difficult to handle within gradient-based optimization frameworks. \oprocend
\end{remark}
% Finally, we note that one could alternatively attempt to enforce stability by requiring that the spectral abscissa of $A-BK$ is negative.  However, this leads to a non-smooth constraint, which makes it difficult to use gradient-based optimization methods.
%
% \marginJC{This last observation could be in a remark?}
%

\section{Algorithmic Implementation}

In this section, we propose an algorithm to solve~\eqref{eq:non-convex-reformulation}.
We start by rewriting~\eqref{eq:non-convex-reformulation} in the standard constrained optimization form
\begin{subequations}
\begin{align}
    &\min\limits_{\theta\in\real^d} \Lc(\theta) \\
    &\text{s.t.} \ C_i(\theta) \leq 0, \ i\in[q],
\end{align}
\label{eq:constrained-optimization-problem}
\end{subequations}
where $q = 2n+m+1$ and
$C_i(\theta) = -\text{det}(\Delta_{(1:i,\,1:i)}(Q,Y))+\epsilon_0$ for $i\in[n]$, 
$C_i(\theta) = -\text{det}(Q_{(1:i,\,1:i)})+\epsilon_0$ for $i\in\{n+1,\hdots,2n\}$, $C_i(\theta) = -\text{det}(R_{(1:i,\,1:i)} )+\epsilon_0$ for $i\in\{2n+1,\hdots,2n+m\}$ and $C_{2n+m+1}(\theta) = -\alpha_{\theta} + \epsilon_0$. Here, $\epsilon_0>0$ is a small tolerance.

We note that~\eqref{eq:constrained-optimization-problem} is in general a non-convex problem. 
We propose to solve it through the \emph{robust safe gradient flow} (RSGF) method, 
which we adopt from~\cite{PM-AM-JC:25-l4dc,AA-RS-MT:10,AA-JC:24-tac}. Specifically, given $\delta, \gamma >0$,
RSGF generates iterates $\{\theta_k\}_{k\in\mathbb{N}}$ as follows:
%
%\marginJC{$1/2h$ is $\frac{1}{2\delta}$?}
%
\begin{align}\label{eq:rsgf}
\notag
    &\theta_{k+1} = p(\theta_k) \\
     &p(\theta):= \text{arg}\min\limits_{y\in\real^d} \nabla \Lc(\theta)^\top (y-\theta) + \frac{1}{2\delta}\norm{y-\theta}^2 \\
    \notag
    &\text{s.t.} \ \gamma \delta C_i(\theta) + \nabla C_i(\theta)^\top (y-\theta) + \frac{1}{2\delta}\norm{y-\theta}^2 \leq 0, \ i\in[q].
\end{align}
%
% \marginJC{I think we should use the name "robust safe gradient flow" and say that we borrow it here. }
%

As shown in~\cite[Lemma 1]{PM-AM-JC:25-l4dc},
under a suitable choice of the parameters $\delta, \gamma$,
the iterates defined by~\eqref{eq:rsgf} converge to a KKT point of~\eqref{eq:constrained-optimization-problem}.
% which are exactly the points $\theta$ satisfying $p(\theta) = \theta$.
An additional benefit of Algorithm~\ref{eq:rsgf} compared to other constrained optimization algorithms such as interior-point~\cite{JN-SW:06} or primal-dual methods~\cite{DF-FP:10} is that the iterates of~\eqref{eq:rsgf} are guaranteed to satisfy the constraints of~\eqref{eq:constrained-optimization-problem} (cf.~\cite[Lemma 1]{PM-AM-JC:25-l4dc}).
This is particularly relevant in the context of problem~\eqref{eq:non-convex-reformulation}, where this implies that the algorithm iterates generate stabilizing nominal controllers and positive slopes of the class-$\mathcal K_\infty$ function.
% even if the algorithm is terminated after a finite number of iterations.

\subsection{Gradient approximation and implementation}

Unfortunately, in the setting considered here, the function $\Lc(\theta)$ and its gradient $\nabla \Lc(\theta)$ are difficult to obtain explicitly, since they involve an integral over all possible initial conditions. Instead, in practice we need to approximate $\nabla \Lc(\theta)$ by sampling a few initial conditions, which produces an estimate $\widehat{\nabla \Lc}(\theta)$ of the gradient.
% The cost function $\mathcal{L}(\theta)$ defined in~\eqref{eq:prelim-cost} involves an expectation over initial conditions and depends on the trajectories of the closed-loop system. Since this quantity cannot be evaluated in closed form, we approximate it using Monte Carlo simulation of the closed-loop dynamics. 
Specifically, let $\bx_0^{(i)}$, $i\in\{1,\dots,N\}$, be sampled from the distribution $\xi$.  For each initial condition, we simulate the closed-loop system~\eqref{eq:prelim-closed-loop} under the parameterized safety filter~\eqref{eq:prelim-parameterized-filter} over a finite horizon $T$ using a time discretization with step size $\Delta t$. 
Let $\bx_t^{(i)}$ denote the state at time $t\Delta t$ and $\theta_k$ be the current parameter, we compute the control $ u_t^{(i)}$ using the parameterized safety filter \eqref{eq:prelim-parameterized-filter} and  propagate the dynamics through a forward Euler discretization: $\bx_{t+1}^{(i)} = \bx_t^{(i)} + \Delta t \big(f(\bx_t^{(i)}) + g(\bx_t^{(i)})u_t^{(i)}\big)$.
The cost $\mathcal{L}(\theta)$ is then approximated as
\begin{align}
    \widehat{\Lc}(\theta):= \frac{1}{N} \sum_{i=1}^N \Big( 
    \phi(\bx_K^{(i)}) + \lambda \sum_{t=0}^{K-1} \psi(\bx_t^{(i)}) \Delta t
    \Big),
\end{align}
where $K = T / \Delta t$. Now, the mapping $\theta \mapsto \widehat{\mathcal{L}}(\theta)$ is defined through these closed-loop rollouts and an estimate of the gradient $\nabla \mathcal{L}(\theta)$ can be obtained by directly differentiating the mapping $\widehat{\mathcal{L}}(\theta)$.

Next, replacing $\nabla \mathcal{L}(\theta)$ with $\widehat{\nabla\mathcal{L}}(\theta)$ in \eqref{eq:rsgf} yields the following iterative algorithm 
\begin{align}\label{eq:estimate-rsgf}
\notag
    &\theta_{k+1} = \hat{p}(\theta_k) \\
    &\hat{p}(\theta) := \text{arg}\min\limits_{y\in\real^d} \widehat{\nabla \Lc}(\theta)^\top (y-\theta) + \frac{1}{2\delta}\norm{y-\theta}^2 \\
    \notag
    &\text{s.t.} \ \gamma \delta C_i(\theta) + \nabla C_i(\theta)^\top (y-\theta) + \frac{1}{2\delta}\norm{y-\theta}^2 \leq 0, \ i\in[q].
\end{align}

% \red{[ED: In the definition of the RSGF update, the current iterate and the optimization variable should be used consistently. In particular, all gradients and constraint values should be evaluated at the current iterate $\theta_k$, and the proximal terms should involve $(y-\theta_k)$ throughout. I think it is just a typo, but please check.]}

We note that~\eqref{eq:estimate-rsgf} is feasible for 
values of $\theta$ in the feasible set, since in that case $y=\theta$ is feasible for~\eqref{eq:estimate-rsgf}.
The following result characterizes the convergence and constraint satisfaction properties of~\eqref{eq:estimate-rsgf}.

\begin{proposition}\longthmtitle{Constraint satisfaction and convergence}\label{prop:constraint-satisfaction-convergence}
    Suppose that for each $\theta\in\real^d$, 
    $\norm{ \widehat{\nabla \Lc}(\theta) - \nabla \Lc(\theta) } \leq \epsilon$,
    for some $\epsilon > 0$.
    Further suppose that $\Lc$ is $L$-smooth and lower bounded by a constant $\Lc_{\text{min}}$,
    and $C_i$ is $L_i$-smooth for each $i\in[q]$.
    Assume also that 
    \begin{align}\label{eq:design-parameter-conditons}
        0 < \delta < \min\{ \frac{1}{\gamma}, \frac{1}{L}, \frac{1}{L_1}, \hdots, \frac{1}{L_q}  \}.
    \end{align}
    %
% \marginJC{Is there something we can say if we start infeasible? B/c we are only evaluating a few trajectories, then the cost is always finite, so RSGF will bring us back to feasible, no? I think that's an important point -- we actually use it later in the sims.}
    %
    %
    Then, if the algorithm is initialized with 
    $\theta_0\in\mathcal{F}:=\setdef{\theta\in\real^d}{C_i(\theta) \leq 0,\ \forall i\in[q]}$, the sequence of iterates generated by~\eqref{eq:estimate-rsgf} satisfies:
    \begin{enumerate}
        \item\label{it:safety} $\theta_k\in\mathcal{F}$ for all $k\in\mathbb{N}$;
        \item\label{it:convergence} Let $K^*\in\mathbb{N}$ be such that $K^* \geq \frac{L(\Lc(\theta_0)-\Lc_{\text{min}})}{4\epsilon^2}$. Then, there exists $k\in[K]$ such that $\norm{\hat{p}(\theta_k)-\theta_k} \leq \frac{4\epsilon}{L}$.
    \end{enumerate}
    Alternatively, if $c>0$ is such that~\eqref{eq:estimate-rsgf} is feasible in the set $\mathcal{F}_c:=\setdef{\theta\in\real^d}{C_i(\theta) \leq c, \ \forall i\in[q]}$ and $\theta_0\in\mathcal{F}_c$, then the sequence of iterates converges to $\mathcal{F}$.
\end{proposition}

\begin{proof}
    Let us first prove item~\ref{it:safety}, which follows an argument analogous to that of~\cite[Lemma 1 (i)]{PM-AM-JC:25-l4dc}.
    Let $i\in[q]$ and $k\in\mathbb{N}$.
    Since $C_i$ is $L_i$-smooth, by~\cite[Lemma 1.2.3]{YN:18} we have 
    \begin{align}\label{eq:Ci-Li-smoothness}
        C_i(\theta_{k+1}) \! \leq \! C_i(\theta_k) \! + \! \nabla C_i(\theta_k)^\top (\theta_{k+1} \! - \! \theta_k) \! + \! \frac{L_i}{2}\norm{\theta_{k+1}-\theta_k}^2.
    \end{align}
    Now, by definition, $\theta_{k+1}$ satisfies
    \begin{align}\label{eq:Ci-constraint}
        \gamma \delta C_i(\theta_k) + \nabla C_i(\theta_k)^\top (\theta_{k+1}-\theta_k) + \frac{1}{2\delta}\norm{\theta_{k+1}-\theta_k}^2 \leq 0.
    \end{align}
    By using~\eqref{eq:Ci-constraint} in~\eqref{eq:Ci-Li-smoothness} we get
    \begin{align}\label{eq:Ci-safety-conclusion}
        C_i(\theta_{k+1}) \leq (1-\gamma\delta) C_i(\theta_k) -\frac{1}{2}\Big( \frac{1}{\delta} - L_i \Big) \norm{\theta_{k+1}-\theta_k}^2.
    \end{align}
    Since $\delta < \frac{1}{\gamma}$ and $\delta < \frac{1}{L_i}$,~\eqref{eq:Ci-safety-conclusion} implies that if $C_i(\theta_k) \leq 0$, then $C_i(\theta_{k+1}) \leq 0$. By induction, and since this argument is valid for all $i\in[q]$, this implies that item~\ref{it:safety} holds.
    
    %
% \marginJC{I wouldn't mind spelling it out briefly, since we have space.}
    %
    Let us now show item~\ref{it:convergence}.
    Since $\Lc$ is $L$-smooth, we have 
    \begin{align*}
        \Lc(\theta_{k+1}) \leq \Lc(\theta_k) + \nabla \Lc(\theta_k)^\top (\theta_{k+1}-\theta_k) + \frac{L}{2}\norm{\theta_{k+1}-\theta_k}^2.
    \end{align*}
    By the optimality conditions of~\eqref{eq:estimate-rsgf} (cf.~\cite[Thm. 12.3]{JN-SW:06}), we have that any feasible $y\in\real^d$ satisfies 
    \begin{align*}
        (y-\hat{p}(\theta_k))^\top ( \widehat{\nabla \Lc}(\theta_k) + \frac{1}{\delta}( \hat{p}(\theta_k) - \theta_k ) ) \geq 0.
    \end{align*}
    %It is even more clearly stated in "Anytime solvers for variational inequalities: The (recursive) safe
    %monotone flow" Section 2.2, which shows that constrained optimization problems can be written as variational inequalities, for which the above inequality is the definition of optimality.
    In particular, $y = \theta_k$ is feasible (thanks to item~\ref{it:safety}, because $\theta_k\in\mathcal{F}$, which means that $C_i(\theta_k) \leq 0$ for all $i\in[q]$ and $k\in\mathbb{N}$) and yields
    \begin{align*}
        \Lc(\theta_{k+1}) - \Lc(\theta_k) &\leq (\nabla \Lc(\theta_k) - \widehat{\nabla \Lc}(\theta_k) )^\top (\theta_{k+1}-\theta_k) \\
        &-\Big( 
        \frac{1}{\delta}-\frac{L}{2}
        \Big) \norm{ \theta_{k+1}-\theta_k }^2.
    \end{align*}
    Now, after using the Cauchy-Schwartz inequality, the bound $\norm{ \widehat{\nabla\Lc}(\theta) - \nabla\Lc(\theta) } \leq \epsilon$, and $\delta < \frac{1}{L}$, we get that whenever $\norm{\hat{p}(\theta_k) - \theta_k} > \frac{4\epsilon}{L}$, 
    \begin{align}\label{eq:Lc-decrease}
        \Lc(\theta_{k+1}) - \Lc(\theta_k) < -\frac{4\epsilon^2}{L} < 0.
    \end{align}
    Now, since $\Lc \geq \Lc_{\text{min}}$, if $K^*\in\mathbb{N}$ was such that $K^* \geq \frac{(\Lc(\theta_0)-\Lc_{\text{min}}) L}{ 4\epsilon^2 }$ and
    $\norm{\hat{p}(\theta_k)-\theta_k} > \frac{4\epsilon}{L}$ for all $k\in[K^*]$, from~\eqref{eq:Lc-decrease}, we would get that $\Lc(\theta_{K^*+1}) < \Lc_{\text{min}}$, which is a contradiction.
    Therefore, there exists $k\in[K^*]$ such that $\norm{\hat{p}(\theta_k)-\theta_k} \leq \frac{4\epsilon}{L}$.

    Finally, let us show that if that~\eqref{eq:estimate-rsgf} is feasible in $\mathcal{F}_c$ and $\theta_0\in\mathcal{F}_c$, then the sequence of iterates converges to the feasible set of~\eqref{eq:constrained-optimization-problem}.
    %
% \marginJC{We say here in finite time, but in the next sentence, we say could be asymptptically. Finite time would be better. And if it is finite time, we should state it in the proposition.}
    %
    This follows from~\eqref{eq:Ci-safety-conclusion}, by noting that for any $i\in[q]$, if $C_i(\theta_k) > 0$, then $C_i(\theta_{k+1}) - C_i(\theta_k) < 0$, which implies that the sequence converges to the feasible set of~\eqref{eq:constrained-optimization-problem} (either asymptotically or in finite time).
\end{proof}

\vspace{0.1cm}

The smoothness assumptions in Proposition~\ref{prop:constraint-satisfaction-convergence} are standard in optimization, and hold if the parameter set is restricted to a compact set.
From~\eqref{eq:Ci-safety-conclusion}, we can see that the choice of $\delta$ and $\gamma$ affects the maximum rate with which the constraint values can approach the boundary of the safe set. Indeed, smaller values of $\delta$ decrease the second term in~\eqref{eq:Ci-safety-conclusion}, but increase the first term. Smaller values of $\{ L_i \}_{i=1}^q$ allow a wider range of values of $\delta$ satisfying~\eqref{eq:design-parameter-conditons}, which can help adjust such rate. Finally, smaller values of $L$ lead to faster convergence as per Proposition~\ref{prop:constraint-satisfaction-convergence} item~\ref{it:convergence}.

%
% \marginJC{How reasonable/unreasonable are the assumptions of the proposition? Also, can you comment on the ineq that involved $\delta$, $\gamma$, and the $L$'s? The 1st two are design parameters, so that inverse relationship must point to some tradeoff. Similarly, the relationship w/ the $L$'s means that smaller $L$'s give us more choice, meaning?}
%

Importantly, Proposition~\ref{prop:constraint-satisfaction-convergence} in the context of problem~\eqref{eq:non-convex-reformulation} ensures that if the iterates are initialized with a stabilizing nominal controller, then all iterates remain stabilizing, and if it is initialized with a controller that is \textit{close} to being stabilizing, the iterates will converge to a stabilizing one (provided that the problem remains feasible).
In particular, since obtaining a small approximation error $\epsilon$ requires sampling a high number of trajectories, and the property that unstable controllers converge to stable ones is independent of $\epsilon$, 
if the initial nominal controller is not stabilizing, one does not need to sample a high number of trajectories and can instead wait for the parameter iterates to become stabilizing to start sampling a high number of trajectories and eventually lead to a better approximation.

We note that as the error bound $\epsilon$ between 
$\widehat{\nabla \Lc}(\theta)$ and
$\nabla \Lc(\theta)$ vanishes,
the values of $\hat{p}(\theta)$ and $p(\theta)$ coincide 
(under standard constraint qualification conditions on~\eqref{eq:constrained-optimization-problem} cf.~\cite[Thm 5.3]{AVF-JK:85}).
%MFCQ
Since the KKT points of~\eqref{eq:constrained-optimization-problem} can be characterized as the values of $\theta$ for which $\theta = p(\theta)$ (cf.~\cite[Lemma 1]{PM-AM-JC:25-l4dc}), Proposition~\ref{prop:constraint-satisfaction-convergence} item~\ref{it:convergence} ensures that for small $\epsilon$, there exists $k\in[K]$ such that $\theta_k$ lies in a neighborhood of a KKT point of~\eqref{eq:constrained-optimization-problem}.
%
% \marginJC{Here, we discuss importance of Prop, item 2. We should emphasize (even if briefly) the importance of Prop, item 1, in regards to the challenge we discussed at the end of Section IV.}
%

%
% \marginJC{If we discuss this "recover from unstable" feature that the algo here has, we can add (have a remark?) a discussion about how, at the beginning, when we start infeasible, it's not that big of a deal to be using fewer trajectories to approximate (e.g., bad $\epsilon$). Once we're feasible (stable), then one can progressively incorporate more and more trajectories to reduce error (smaller $\epsilon$). Potentially there are other results to prove, essentially connecting $\epsilon$ to the number of trajectories, etc.}
%

We summarize the overall training procedure in Algorithm~\ref{alg:training}.

\begin{algorithm}[tbh]
\caption{Training via Robust Safe Gradient Flow with Trajectory Sampling}
\label{alg:training}
\begin{algorithmic}[1]

\STATE Initialize parameters $\theta_0$
\FOR{$k = 0,1,\dots$}

\STATE Sample $N$ initial conditions $\{\bx_0^{(i)}\}_{i=1}^N \sim \xi$

\FOR{each $i=1,\dots,N$}
    \STATE Set $\bx^{(i)}_0 = \bx_0^{(i)}$
    \FOR{$t = 0,\dots,K-1$}
        \STATE Compute control $ u_t^{(i)}$ via %parameterized safety filter 
        \eqref{eq:prelim-parameterized-filter} with current state  $ x_t^{(i)}$ and parameter $\theta_k$
        \STATE Propagate dynamics:
        \[
        \bx_{t+1}^{(i)} = \bx_t^{(i)} + \Delta t \big(f(\bx_t^{(i)}) + g(\bx_t^{(i)})u_t^{(i)}\big)
        \]
    \ENDFOR
\ENDFOR

\STATE Compute approximated loss $\widehat{\mathcal{L}}(\theta_k)$ and gradient $\widehat{\nabla \mathcal{L}}(\theta_k)$

\STATE Update parameters:
$\theta_{k+1} = \hat{p}(\theta_k)
\quad \text{(cf.~\eqref{eq:estimate-rsgf})}$
\ENDFOR

\end{algorithmic}
\end{algorithm}
\vspace{-0.4cm}
%
% \marginJC{ALgorithm steps have several typos/inconsistencies with main text: $\eta$ was $\xi$ in text, (18) is $\hat{p}$, not $p$,...}
% %

\subsection{Multiple obstacles}\label{sec:multiple-obs}

The proposed framework naturally extends to the case of multiple obstacles. 
Suppose the safe set is defined as the intersection
\[
\Cc = \bigcap_{j=1}^M \Cc_j, 
\quad 
\Cc_j = \{\bx \in \real^n \mid h_j(\bx) \ge 0\}.
\]

A standard approach is to enforce the CBF condition for each barrier function $h_j$ in the safety filter~\eqref{eq:prelim-parameterized-filter}, which leads to a quadratic program with multiple constraints. While this formulation is conceptually straightforward, it becomes computationally demanding in the trajectory-based training setting considered here. Indeed, each rollout requires solving a quadratic program at every time step, and over a time horizon $T$ with discretization step $\Delta t$, a single rollout involves $T/\Delta t$ solves. When $N$ trajectories are sampled to approximate $\mathcal{L}(\theta)$, this results in solving a large number of quadratic programs.  Moreover, in the presence of multiple constraints, the control input is defined implicitly as the optimizer of a constrained quadratic program. As a result, computing the gradient of $\widehat{\mathcal{L}}(\theta)$ with respect to $\theta$ requires differentiating through the solution map of this optimization problem, which further increases computational cost and complexity.

To address these issues, we adopt the same log-sum-exp relaxation as in \cite{TGM-ADA:23} to combine all the CBFs. Specifically, we define 
$    \tilde{h}_\beta(\bx) 
    = -\frac{1}{\beta} \log \Big( \sum_{j=1}^M \exp(-\beta h_j(\bx)) \Big),
$ 
where $\beta > 0$ is a parameter. The following result characterizes this relaxation.

% \red{[ED: is (21) the correct relaxation? It seems to ne that it should be 
% $$
% -\frac{1}{\beta} \log \left( \sum_{j=1}^M \exp(-\beta h_j(\bx)) \right).
% $$
% I think the following: since
% \[
% \Cc = \bigcap_{j=1}^M \{\bx \in \real^n \mid h_j(\bx) \ge 0\}
%     = \{\bx \in \real^n \mid \min_{j\in[M]} h_j(\bx) \ge 0\},
% \]
% the aggregated barrier function must approximate the pointwise minimum of the functions $\{h_j\}_{j=1}^M$. To this end, we use the standard log-sum-exp approximation
% \[
% \tilde{h}_\beta(\bx)
% =
% -\frac{1}{\beta}\log\!\left(\sum_{j=1}^M e^{-\beta h_j(\bx)}\right),
% \qquad \beta>0,
% \]
% which satisfies $\tilde{h}_\beta(\bx)\to \min_{j\in[M]} h_j(\bx)$ as $\beta\to\infty$. Indeed, by the classical bounds for the smooth minimum,
% \[
% \min_{j\in[M]} h_j(\bx) - \frac{\log M}{\beta}
% \le
% \tilde{h}_\beta(\bx)
% \le
% \min_{j\in[M]} h_j(\bx).
% \]
% Therefore, the superlevel set of $\tilde h_\beta$ provides an inner approximation of the safe set $\Cc$. 

% Am I mistaken? 
% ]}

% \red{[ED: is the lemma correct? Do we need the minus sign?]}

\vspace{0.1cm}

\begin{lemma}[{\cite[p.~72]{SB-LV:09}}]
\label{lem:log-sum-exp-bounds}
For any $a_1,\dots,a_M \in \real$ and $\beta > 0$, define
$
\phi_{\beta}(a_1,\dots,a_M)
:=
-\frac{1}{\beta}\log\Big( \sum_{j=1}^M \exp(-\beta a_j) \Big).
$
%
% \marginJC{If we use exp above, better to use exp here too.}
%
It holds that
\[
\min_{j\in[M]} a_j - \frac{\log M}{\beta}
\;\le\;
\phi_{\beta}(a_1,\dots,a_M)
\;\le\;
\min_{j\in[M]} a_j. \quad \Box
\]
\end{lemma}

\vspace{0.1cm}

By Lemma~\ref{lem:log-sum-exp-bounds}, the set 
\[
\hat{\Cc} := \{\bx \in \real^n \mid \tilde{h}_\beta(\bx) \ge 0\}
\]
satisfies $\hat{\Cc} \subseteq \Cc$. Furthermore, as $\beta \to \infty$, $\tilde{h}_\beta(\bx)$ converges to $\min_j h_j(\bx)$, and hence $\hat{\Cc}$ approaches $\Cc$.

Using $\tilde{h}_\beta$ in place of the individual barrier functions results in a safety filter with a single constraint. This yields several important benefits. First, the number of constraints in the safety filter no longer scales with the number of obstacles, which significantly reduces the computational burden during trajectory rollouts. In particular, the control input can be computed without solving a multi-constraint quadratic program at each time step. Second, from~\eqref{eq:v-expression}, the resulting control input admits a closed-form expression. This eliminates the need to solve an optimization problem online and substantially reduces the overall computational cost. Finally, the closed-form expression makes the dependence of the control input on $\theta$ explicit. This avoids differentiating through the solution map of a quadratic program and simplifies the computation of gradients of $\widehat{\mathcal{L}}(\theta)$ with respect to~$\theta$.

% In summary, the log-sum-exp approximation provides a tractable surrogate for multiple barrier constraints, trading a controllable approximation error for improved computational efficiency and differentiability.

\section{Numerical Experiments}

In this section, we illustrate the performance of the proposed approach on several planar examples. 
We consider a single integrator system
\[
\dot{\bx} = A\bx + B\bu=\bu,
\]
with $A = 0$ and $B = I_2$, and use linear nominal controllers of the form $\bu = -K\bx$. 
The safety filter is constructed using the parameterized CBF formulation as in \eqref{eq:prelim-parameterized-filter}.
% \marginED{Perhaps, write the model for the  integrator and then $\dot{\bx} = I_2\bu,$? }

In all experiments, the trajectories are generated using a forward Euler discretization with step size $\Delta t = 0.02$. 
The cost function is chosen as
\[
\mathcal{L}(\theta,\bx_0) = \|\bx(T)\|_2^2 + \lambda \int_0^T \|\bx(t)\|_2^2 dt,
\]
with $\lambda = 0.1$. 
The expectation in~\eqref{eq:prelim-cost} is approximated using Monte Carlo sampling with batch size $512$, where initial conditions are sampled uniformly from the safe set.

% \marginJC{Plots are great. Could we also have markers (e.g., red stars) for the equilibria in each plot? (in these fig and the others) That'd make the discussion in the caption even easier to visually verify.}
%
% \marginJC{Do we want to also mark $0$? But differently from the others (e.g., black star).}

For each experiment, we first choose an initial nominal gain $K_0$ and then construct $(Q_0,Y_0)$ as follows. 
If the matrix $A-BK_0$ is Hurwitz, we let $P_0$ be the positive definite solution of the Lyapunov equation
\[
(A-BK_0)^\top P_0 + P_0(A-BK_0) = -I,
\]
and define
\[
Q_0 = P_0^{-1},
\qquad
Y_0 = K_0 Q_0.
\]
If the matrix $A-BK_0$ is unstable, we simply take $Q_0 = I_2$ and $Y_0 = K_0$. 
%
% \marginJC{This points to my margin in the proposition above regarding what happens if we start infeasible!}
%
In all experiments, the parameter $\alpha$ is initialized as $\alpha_0 = 5$.

% \marginED{Let's make the figures larger. We have 8 pages.}

\subsection{Invariance in a Bounded Safe Set}

We first consider a bounded safe set defined as a disk,
\[
\Cc = \{\bx : r^2 - \|\bx-\bx_c\|^2 \geq 0\},
\]
and the initial conditions are chosen uniformly on a circle inside the disk. 
By \cite{PM-YC-EDA-JC:25-jnls},
%
% \marginJC{Would be nice to refer to specific results in this section when we mention obstructions: like \cite[Theorem~BLA]{PM-YC-EDA-JC:25-jnls}. Same margin for other subsections here.}
% %
in this setting, the best scenario that one can expect is the absence of undesired equilibria, and that $\Cc$ coincides with the region of attraction of the origin~\cite[Proposition 4.10]{PM-YC-EDA-JC:25-jnls}.

In Fig.~\ref{fig:case1} (left), we show the trajectories obtained with the safety filter applied to the initial controller. 
There are two undesired equilibria located on the boundary of $\Cc$, one of which is asymptotically stable and has a region of attraction with positive measure.  In Fig.~\ref{fig:case1} (right), we show the trajectories obtained after training. No undesired equilibrium is observed in this case. All trajectories remain in $\Cc$ and converge to the origin.

\begin{figure}[htb]
    \centering
    \includegraphics[width=0.49\textwidth]{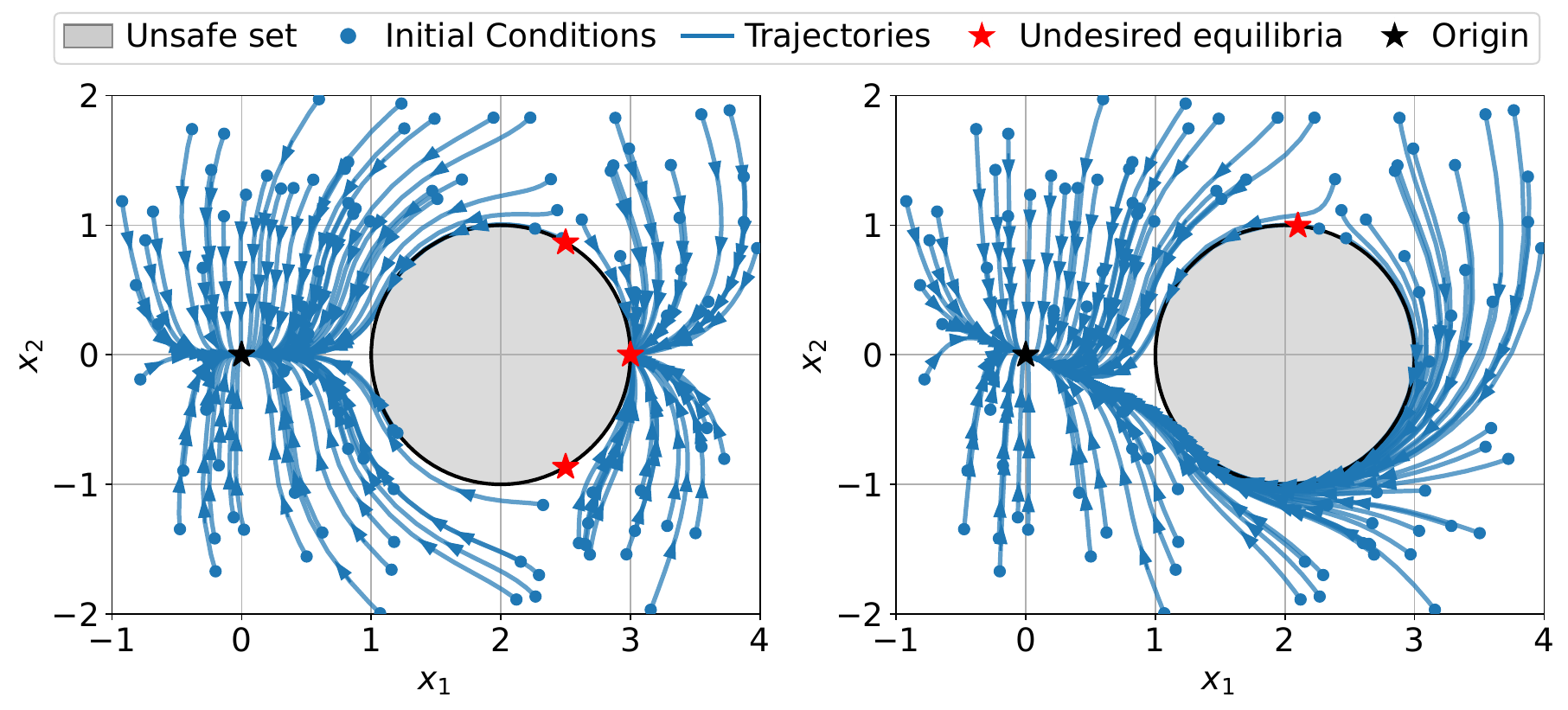}
    \vspace{-0.3cm}
    \caption{State trajectories generated by the initial (left) and trained (right) controllers. Gray regions indicate unsafe sets. Under the initial controller, two undesired equilibria appear on the boundary of the safe set, one of which is asymptotically stable and its region of attraction has a positive measure. After training, no undesired equilibrium is observed, and all trajectories remain within the safe set and converge to the origin. }
    \vspace{-0.3cm}
    \label{fig:case1}
\end{figure}

\begin{figure}[htb]
    \centering    \includegraphics[width=0.49\textwidth]{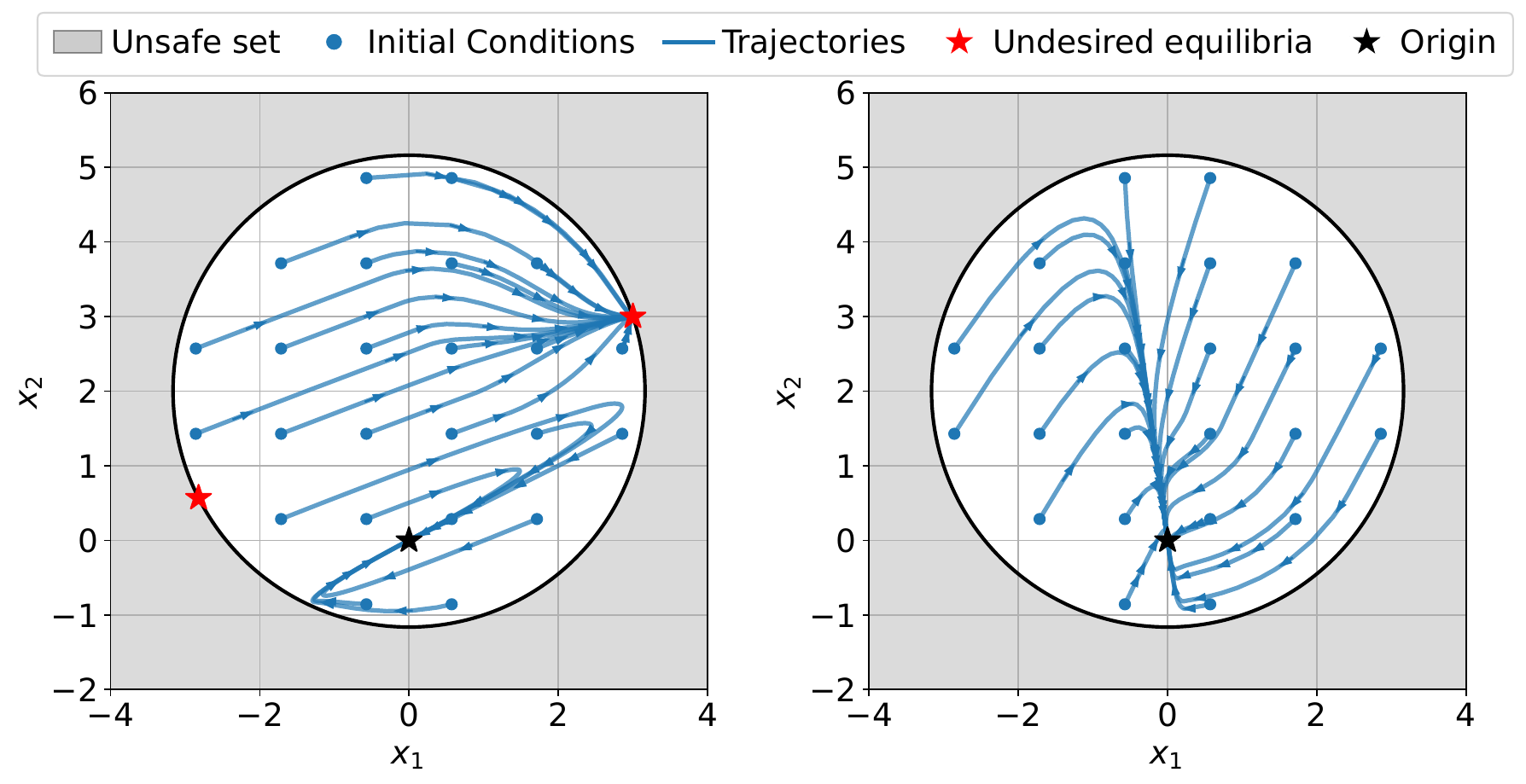}
    \vspace{-0.3cm}
    \caption{State trajectories generated by the initial (left) and trained (right) controllers. 
Gray regions indicate unsafe sets. The initial controller induces an asymptotically stable undesirable equilibrium on the boundary of the obstacle, causing some trajectories to converge to the unsafe set. After training, the asymptotically stable undesirable equilibrium is eliminated, and the resulting controller keeps all trajectories within the safe set while yielding improved convergence behavior.}
    \label{fig:case2}
\end{figure}

\subsection{Avoidance of a Single Obstacle}

We next consider a single circular obstacle centered at $\bx_c$ with radius $r$, and define the safe set as
\[
\Cc = \{\bx : 
r^2 - \|\bx-\bx_c\|^2 \leq 0
% \|\bx - \bx_c\|^2 - r^2 \geq 0
\}.
\]
The initial conditions are chosen uniformly on a circle enclosing the obstacle. 
By \cite{chen2024characterization} and~\cite[Proposition 4.9]{PM-YC-EDA-JC:25-jnls}, for linear nominal controllers with safety filter, we can not prevent the emergence of undesired equilibria, and the best scenario that one can expect is the presence of a single unstable undesired equilibrium.

In Fig.~\ref{fig:case2} (left), we show the trajectories obtained with the safety filter applied to the initial controller. 
There is one undesired equilibrium located on the boundary of the obstacle, which is asymptotically stable and has a region of attraction with positive measure. The trajectories initialized in this region converge to this equilibrium, while all remaining trajectories converge to the origin. After training, as shown in Fig.~\ref{fig:case2} (right), all trajectories remain in $\Cc$ and converge to the origin. 
This behavior indicates the absence of asymptotically stable undesired equilibria. 
Moreover, the trajectory pattern suggests the presence of a single undesired equilibrium on the boundary, which is unstable, in agreement with \cite{chen2024characterization}.

\subsection{Obstacle Avoidance with Multiple Constraints}
% \begin{figure}
%     \centering
%     \includegraphics[width=0.45\linewidth]{Figures/figure5.png}
%     \includegraphics[width=0.45\linewidth]{Figures/figure6.png}
%     \caption{Closed-loop trajectories with multiple obstacles. Left: initial controller. Right: learned controller.}
% \end{figure}

\begin{figure*}[htb]
    \centering
    \includegraphics[width=0.75\textwidth]{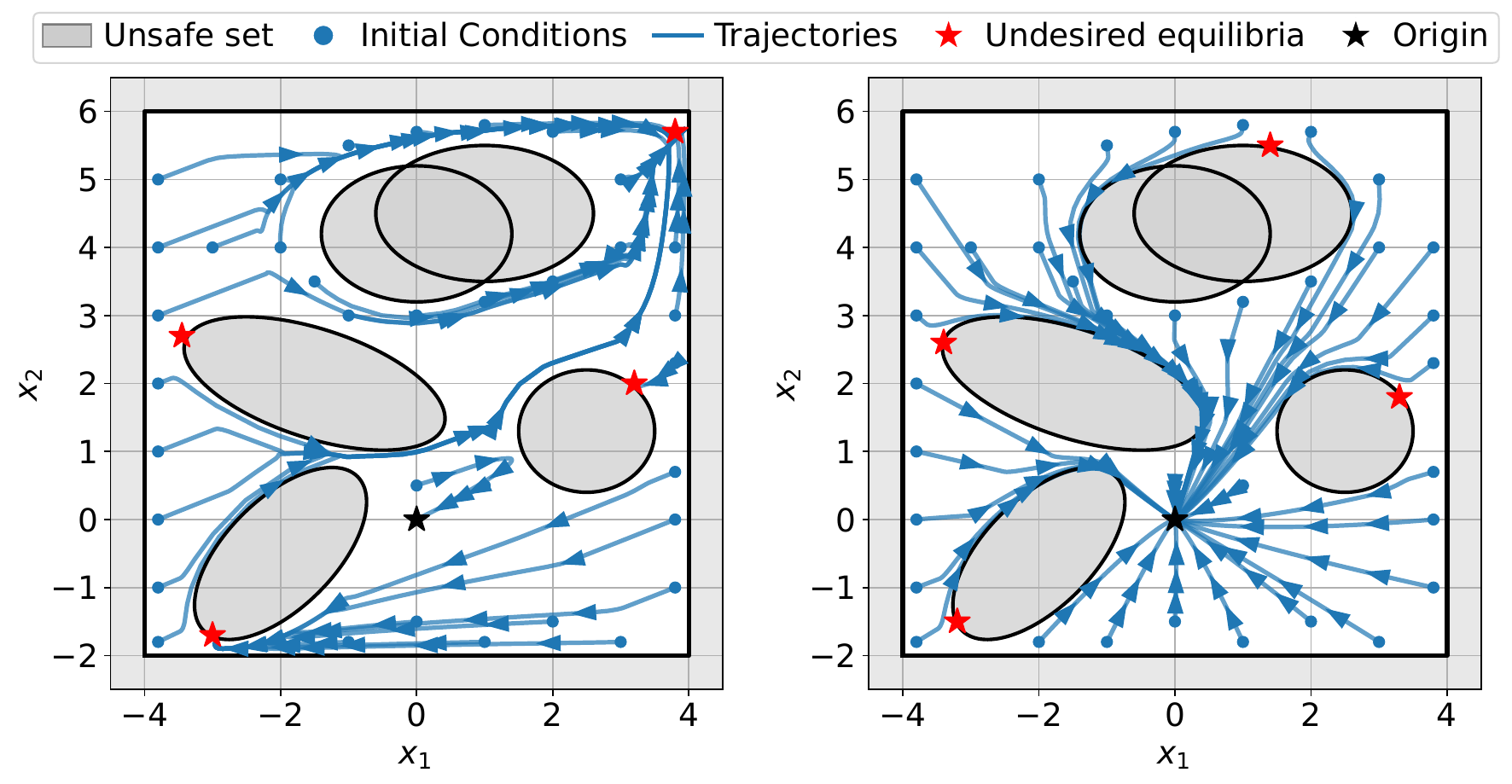}
    \vspace{-0.3cm}
    \caption{State trajectories generated by the initial (left) and trained (right) controllers. 
Gray regions indicate unsafe sets. Given a same set of initial conditions, the trajectories under the trained controller 
converge to the origin while ensuring obstacle avoidance, whereas many trajectories under the initial controller converge to undesirable equilibria located near the top-right corner and on the boundaries of the ellipsoid obstacles.}
\vspace{-0.3cm}
    \label{fig:case3}
\end{figure*}

Finally, we consider a scenario with multiple obstacles. 
The safe set is defined as the intersection of several constraints, including multiple rotated elliptical obstacles and box constraints defining the workspace. 
Specifically, the safe set is given by
\begin{align*} 
\Cc = \bigcap_{j=1}^{5} \{\bx : (\bx - \bc_j)^\top R_j^\top D_j R_j (\bx - \bc_j) - 1 \ge 0\}
\;\bigcap \;\\
\{\bx :  x_{\min} \le x_1 \leq x_{\max} ,\;
y_{\min} \le x_2 \leq y_{\max}\},
\end{align*}
where $\bc_j \in \real^2$ is the center of the ellipse, 
$R_j \in \real^{2\times 2}$ is a rotation matrix, and 
$D_j = \mathrm{diag}(1/a_j^2,\, 1/b_j^2)$ defines the shape of the ellipse with semi-axes $a_j$ and $b_j$.
To reduce the computational complexity of the safety filter, we use the log-sum-exp approximation of Section~\ref{sec:multiple-obs} to construct a single barrier function, with smoothing parameter $\beta = 5$. 
The initial conditions are sampled uniformly from the safe region.

Fig.~\ref{fig:case3} (left) shows the trajectories obtained with the safety filter applied to the initial controller. Three asymptotically stable undesired equilibria and one unstable undesired equilibrium are observed. 
The asymptotically stable undesired equilibria are located near the top-right corner, as well as on the boundaries of the bottom-left and middle-right obstacles.  After training, as shown in Fig.~\ref{fig:case3} (right), 
 undesired equilibria still exist, but all of them are unstable.
All trajectories corresponding to the sampled initial conditions remain within $\Cc$ and converge to the origin. This result demonstrates that the proposed approach effectively eliminates stable undesired equilibria in complex multi-obstacle environments, while preserving safety and ensuring convergence to the desired equilibrium.

% ancla
\section{Conclusions}
This paper studied the design of safety-filtered controllers with improved dynamical properties for control-affine systems. 
We proposed a policy optimization framework that jointly parameterizes the nominal controller and the safety filter, and optimizes them using trajectory-based objectives while enforcing Lyapunov-based stability conditions. To ensure that the stability conditions are satisfied throughout the learning process, we incorporate robust safe gradient flow with estimated policy gradient. 
The resulting method enables systematic exploration of controller parameters while maintaining both safety and stability guarantees. The numerical results demonstrate that the proposed approach effectively reduces undesirable equilibria and improves convergence properties, while preserving forward invariance of the safe set. 
Future work will focus on extending the proposed framework to more general nonlinear systems and on the design of parameterized nonlinear nominal controllers.

%
% \marginJC{Update [3], is "to appear"!}
%
\bibliography{alias,Main-add,JC,reference}
\bibliographystyle{IEEEtran}

\end{document}